\documentclass[11 pt, book]{amsart}
\usepackage{amsmath}
\usepackage{latexsym}
\usepackage[all]{xy}
\usepackage{color}
\newtheorem{teorema}{Theorem}[section]
\newtheorem{definicion}[teorema]{Definition}
\include{amslatex}
\newtheorem{proposicion}[teorema]{Proposition}

\newtheorem{comentario}[teorema]{Remark}

\newtheorem{ejemplo}[teorema]{Example}

 \numberwithin{equation}{section}

\begin{document}
\begin{title}[General theory of non-reversible local dynamics]
{General theory of non-reversible local dynamics}
\end{title}
\clearpage\maketitle
\thispagestyle{empty}
\begin{center}
by
\end{center}
\bigskip
\begin{center}
\author{Ricardo Gallego Torrom\'e}
\end{center}
\begin{center}
\address{Department of Mathematics\\
Faculty of Mathematics, Natural Sciences and Information Technologies\\
University of Primorska, Koper, Slovenia}
\end{center}
\begin{abstract}
A general theory of dynamics is formulated with the aim of its application in emergent quantum mechanics. In such a framework it is argued that the fundamental dynamics of emergent quantum mechanics must be non-reversible.
\end{abstract}
\section{Introduction}
The nature and significance of time in quantum gravity is a fundamental problem upon which resolution partially lies the possibility of finding an unifying theory embracing both, quantum theory and the gravitational interaction. Any theory aimed to supersede quantum mechanics and general relativity needs to address the question of time. The current dominant schemes like loop quantum gravity or string theory, neglect the ontological character of time. Instead, in the development of hard quantum theories of gravity, where gravitation is described by a quantum field, they adopt a timeless block universe type description of physical systems, in consistence with the prevail general relativistic viewpoint.

This attitude, more than driven by experience, is driven by conviction, namely that quantum mechanics is fundamental and gravitation must be described by a quantum theory. The time-free block universe appears more as a consequence of the schemes. As such, it is again a consequence of the choice of other assumptions of the theory.

Recently, the possibility that gravity is also emergent has been discussed \cite{Ricardo2015a,Ricardo 2019a} from the  point of view a theory of emergent quantum mechanics  \cite{Ricardo06, Ricardo2014}. In the new approach to the foundations of quantum theory, all quantum degrees of freedom are associated with a coarse grained description of a fundamental dynamics with fundamental degrees of freedom. Although the fundamental dynamics is unspecified in the details, in order that the quantum mechanical description of dynamical systems arises as an emergent phenomena, the fundamental dynamics must possess several key properties. Among the properties that must definitely have, the new dynamics should be either a non-reversible dynamics or a reversible one. Therefore, a general theory of dynamics capable to address such question is needed, with the aim on the applicability  to emergent quantum dynamics.

Starting from general assumptions and integrating further hypotheses in the logical system when required, we discuss in this paper a general theory of dynamics. In this context, we observe that the most natural solution to the problem of reversibility versus non-reversibility falls in the side of non-reversibility.
The general theory of dynamics presented in this paper is based upon three notions: a general concept of time parameters as subsets of certain type of number fields, a general notion of the mathematical objects that change with time as sections of sheaves and a general form of dynamics as a flow rule on a geometric space of objects. In such a framework, a notion of non-reversibility is introduced and the above mentioned question on reversibility/non-reversibility is discussed. We show that Finsler geodesic flows are in general non-reversible and that quantum dynamical laws, despite they can be $T$-inversion invariant, are non-reversible as well. We argue that the fundamental non-reversibility discussed in our general theory of dynamics provides an arrow of time that must coincide, at least locally, with the thermodynamic arrow of time that we observe. Finally, we apply the general theory, inferring that it is more natural to assume non-reversibility than reversibility for a generic dynamics and in particular, for a fundamental dynamics. This conclusion applies to the dynamics of emergent quantum mechanics.

\section{General notion of dynamical system}
In the attempt to provide a consistent theory that supersedes the standard quantum mechanical description of Nature, one should be careful to do impose formal conditions driven by our limited experience at phenomenological  levels but that are not necessary at the fundamental level. The methodology that can help to achieve this objective is to construct a general theory of dynamics, driven by mathematical natural schemes. However, there are dangers with such a strategy, since an abusively too general theory of dynamics could bring to the mathematical framework null capacity of constraining the content and null strength to shape the possible physical laws.

A general theory of dynamics should necessarily contain a notion of time parameter, a well defined notion of mathematical objects that evolve and address what is a dynamical law. These three elements cannot be arbitrary and indeed, they are interconnected. We address them in this section, starting by explaining the notion of time parameter we will adopt.

\subsection{Notion of time parameter}
It is natural to start a general theory of dynamics by clarifying the notion of {\it time parameters} that can be used in the description of the dynamics. Given a number field
 $(\mathbb{K},+,\cdot)$, where there is an additive group operation
 $+:\mathbb{K}\times \mathbb{K}\to \mathbb{K}$
  on $\mathbb{K}$ and a multiplicative operation
  $\cdot:\mathbb{K}\times\mathbb{K}\to \mathbb{K}$ that determines a multiplicative abelian group on $\mathbb{K}^*=\,\mathbb{K}\setminus\{0\}$,
 the time  parameters that we will consider are subsets $J\subset\,\mathbb{K}$ subjected to the following restrictions:
 \begin{enumerate}

 \item The binary operation of addition $+:\mathbb{K}\times \mathbb{K}\to \mathbb{K}$ is required to be inherited by $J$ in the form of the operation $+:J\times J\to \mathbb{K}$. If the dynamics is associated with a flow law, then this requirement is a necessary condition.

 \item In order to define non-linear expressions with dynamical properties, it is necessary that the product of arbitrary elements of $J\subset\,\mathbb{K}$ to be well defined.

 \item  In order to define a notion of incremental quotients or derivative operation limits, it is required the existence of the inverse elements $(t_2-t_1)^{-1}\in\,\mathbb{K}$ for elements $t_2,t_1\in \,J$ close enough in an appropriate sense, except perhaps when $t_2=t_1$.

 \end{enumerate}
 For the purposes of the last point above, it is necessary to endow $J\subset \,\mathbb{K}$ with a quasi-metric function with values on $\mathbb{R}$, namely a function  $d_{\mathbb{K}}:\mathbb{K}\times \mathbb{K}\to \mathbb{R}$ satisfying the metric axioms except the symmetric axiom. The general notion of quasi-metric function is the following:
 \begin{definicion} Let ${\bf T}$ be a set and $\mathbb{K}'$ an ordered number field.
A quasi-metric is a function $\varrho:{\bf T}\times {\bf T}\to \mathbb{K}'$ such that
    \begin{enumerate}
    \item $\varrho(u,v)\geq 0$, for each $u,v\in\, {\bf T}$,
    \item $\varrho(u,v)=0$, iff $u=v\in \,{\bf T}$,
    \item $ \varrho(u,w)\leq \varrho(u,v)+\,\varrho(v,w)$, for each $u,v,w\in \,{\bf T}$.
    \end{enumerate}
\label{quasi-metric}
\end{definicion}
This notion of quasi-metric function is analogous to the one found in the literature \cite{Javaloyes et al.,Wilson1931}, except that we have substituted the real number field $\mathbb{R}$ by a generic ordered number field $\mathbb{K}'$.
In order to be able to formulate the triangle inequality for $d_{\mathbb{K}}$, the number field  $\mathbb{K}'$ must be an ordered field. Furthermore, a quasi-distance function determines a notion of (non)-symmetric distance on ${\bf T}$: the distance between two points $u,v\in\,{\bf T}$ is $d_{\bf T}(u,v):=\,\varrho(u,v)$.

Let us remark that, in order to have a well defined incremental quotient, it is enough to have defined an quasi-metric function of the form
\begin{align}
d_{\mathbb{K}}:\mathbb{K}\times \mathbb{K}\to \mathbb{K}'
\label{distance function in the field body}
\end{align}
From such a function $d_{\mathbb{K}}$, $J\subset\mathbb{K}$  inherits a quasi-metric function $d_J:\,J\times J\to \mathbb{K}'$.
Given a quasi-metric function $d_{\mathbb{K}}$,  there are two possibilities to define incremental quotients,
\begin{itemize}

\item There is a non-zero minimal distance $d_{\mathbb{K}min}>0$ such that if $t_1\, \neq t_2$, then $d(t_1,t_2)\geq d_{\mathbb{K}min}>\,0$. In this case, the incremental quotient of a function $\psi:J\to \mathcal{E}$ is defined by
\begin{align*}
\frac{d\psi}{dt}:=\,\frac{1}{t_2-t_1}\,\left(\psi(t_2)-\,\psi(t_1)\right)
\end{align*}
such that $d_{\mathbb{K}}(t_2,t_1)=\,d_{\mathbb{K}min}$.

\item In the case when $d_{\mathbb{K}min}=0$, as for instance it happens for the fields $\mathbb{Q}$, $\mathbb{R}$ or $\mathbb{C}$ endowed with the usual notion of distance, the incremental quotient limit of a map $\psi:J\to \mathcal{E}$, when it exists, is defined by the expression
\begin{align*}
\frac{d\psi}{dt}=\lim_{d_{\mathbb{K}}(t_2,t_1)\to\, 0}\,\frac{1}{t_2 -\,t_1}\,\left(\psi(t_2)-\psi(t_1)\right).
\end{align*}

\end{itemize}
Both possibilities can be unified in the form of a single notion:
\begin{align}
\frac{d\psi}{dt}:=\,\lim_{d_{\mathbb{K}}(t_2,t_1)\to \,d_{\mathbb{K}min}}\,\frac{1}{t_2 -\,t_1}\,\left(\psi(t_2)-\psi(t_1)\right),
\label{incremental quotient}
\end{align}
where $d_{\mathbb{K}min}=0$ for continuous number fields and it is assumed that the space $\mathcal{E}$ where $\psi(t)$ takes values allows to define $\psi(t_2)-\psi(t_1)$ and to take the incremental limit.

Let $\mathbb{K}$ be a number field endowed with a distance function $d_\mathbb{K}:\mathbb{K}\times \mathbb{K}\to \mathbb{K}'$. If $d_{\mathbb{K}min}\neq 0$, then $t_2-t_1$ is considered to be  {\it small} if $d_{\mathbb{K}}(t_2,t_1)=\,d_{\mathbb{K}min}$. If $d_{\mathbb{K}min}=0$, then $t_2-t_1$ is small if for all purposes when considering the relevant limit expression, the difference $(t_2-t_1)$ can be approximated by the neutral element $0\in\,\mathbb{K}'$.
\begin{comentario}
The formal relation $(t_2-t_1)\equiv 0\in\,\mathbb{K}$ could be interpreted in several ways. One way to understand it is in the context of certain limit expressions as in the evaluation of incremental quotients.

Also, from the above discussion it is natural to assume that the number field $\mathbb{K}$ is endowed with a quasi-distance function $d_{\mathbb{K}}$ and that such quasi-distance function is continuous. Otherwise, the absence of continuity in the quasi-distance function $\mathbb{K}$ could imply that the increasing quotients \eqref{incremental quotient} depend upon the sequence of elements $t_2$ in the neighborhood of $t_1$ in an unnatural way because of the particularities of the distance function $d_\mathbb{K}$.
\label{cometarion on standard and non-standard analysis}
\end{comentario}

Let us emphasize the following points:
\begin{itemize}
\item For the purposes of defining a dynamical theory, the field $\mathbb{K}$ does not need to be complete.

\item Although the number field $\mathbb{K}$ does not need to be commutative, but for mathematical convenience, we will assume that $d_{\mathbb{K}}$ is commutative.

\item Archimedean versus non-archimedean numerical fields. If the field $\mathbb{K}$ is ordered and contains the field of rational numbers $\mathbb{Q}$, one can consider if the archimedean property holds on $\mathbb{K}$. At the moment, there is no restriction on the assumption of the archimedean property. This point should be considered a technical point.
\end{itemize}

The above considerations suggest that a general notion of time parameter is provided by the following,
\begin{definicion}
Let $(\mathbb{K},+,\cdot,d_\mathbb{K})$ be an ordered number field equipped with a continuous quasi-distance function $d_\mathbb{K}:\mathbb{K}\times \mathbb{K}\to \mathbb{K}'$. A time parameter is a subset $J\subset \,\mathbb{K}$ such that
\begin{enumerate}
\item $(J,+)$ is a sub-group of $(\mathbb{K},+)$,

\item For any $t_1\in J$, there are elements $t_2\in\,J$ such that  $(t_2-t_1)\in\,\mathbb{K}$ is small.
\end{enumerate}
\label{definicion of time parameter}
 \end{definicion}
If $(J,+)\,=\,\mathbb{K}$ we say that the time parameter with values in $J$ is complete.

Let us consider a change of time parameter or {\it time re-parametrization}, that in the above setting, can be defined as follows. Consider a field automorphism $\theta:\mathbb{K}\to \mathbb{K}$. A change of parameter is a restriction $\theta|_J:J\to \mathbb{K}$. It is direct that the parameter $\theta(J)\subset \mathbb{K}$ is consistent with three conditions, namely with the definition of group law, with the possibility to combine the new parameter with other variables to provide non-linear expressions and with the notion of incremental quotient. In the case that $(J,+)$ is complete, then $\theta$ is an automorphism of $\mathbb{K}$.

\begin{ejemplo}
Typical examples of time parameters arise when $\mathbb{K}$ is the field of real numbers $\mathbb{R}$. But the field number $\mathbb{K}$ could be a discrete field such as the field of rational numbers $\mathbb{Q}$ or algebraic extensions of $\mathbb{Q}$. In both cases, the incremental quotients are well defined, with a minimal distance $d_{\mathbb{K}min}=0$.
\end{ejemplo}
\begin{ejemplo}
An example of finite number field that can be used to define time parameters for dynamics is the prime field of class $[k]$ module $p$ with $p$ prime,
\begin{align*}
\mathbb{Z}_p :=\{[0],[1],[2],[3],...,[p-1]\}.
\end{align*}
 In this case, there is a natural distance function $d_{\mathbb{Z}_p}:\mathbb{Z}_p\times \mathbb{Z}_p\to \,\mathbb{R}$ defined by the expression
\begin{align}
d_{\mathbb{Z}_p}([n],[m])=\,\{|n_0-m_0|,\,n_0\in\,[n],\,m_0\in\,[m],\,0\leq n_0,m_0\leq p-1\},
\label{distance function in Zp}
\end{align}
 which is a continuous function with the discrete topology of $\mathbb{Z}_p$.

 The following properties can easily be proved:
 \begin{enumerate}

  \item The induced topology in $\mathbb{Z}_p$ from the distance function $d_{\mathbb{Z}_p}$ and the discrete topology of $\mathbb{Z}_p$ coincide.

\item The minimal distance function for $d_{\mathbb{Z}_p}$ is $d_{\mathbb{Z}_p min}=1$.
It follows by an induction argument that the possible subsets $J\subset \mathbb{Z}_p$ that can be used as time parameters, according to definition \ref{definicion of time parameter}, coincide with $\mathbb{Z}_p$ itself. Effectively, if $[t_1]\in\,J\subset\,\mathbb{Z}_p$, then there is (by point 2. in definition \ref{definicion of time parameter}) another $[t_2]\in J$ such that $[t_2]=\,[t_1+1]=\,[t_1]+\,[1]$. Since $card (\mathbb{Z}_p )=\,p$  is finite, then it follows the result by repeating the argument.

\item $\mathbb{Z}_p$ with the distance topology induced from $d_{\mathbb{Z}_p}$ is Haussdorff separable. For this, given two points $[k_1]\neq \,[k_2]\in\,\mathbb{Z}_p$ it is enough to consider the balls
    \begin{align*}
    B([k_i],1/4):=\,\{[k]\in\,\mathbb{Z}_p\,s.t.\,d_{\mathbb{Z}_p}([k],[k_i])<1/4\},\,i=1,2
    \end{align*}
    of radius $1/4$.
    Then $B([k_1],1/4)=\,\{[k_1]\}$, $B([k_2],1/4)=\{[k_2]\}$  and are such that $B([k_1],1/4)\cap B([k_2],1/4)=\emptyset$.

\end{enumerate}

Furthermore, $\mathbb{Z}_p$ can be endowed with an order relation,
  \begin{align*}
 [n] <\,[m]\,\textrm{iff}\,\,\, n_0\,<m_0,\,\,\textrm{with}\, \,n_0\in\,[n],\,m_0\in\,[m],\,0\leq n_0,m_0\leq p-1.
 \end{align*}
According to out discussion above, $\mathbb{Z}_p$ can serve as the set time parameter for a dynamics.

  \end{ejemplo}

\subsection{Notion of configuration space and associated mathematical objects defined over it}
The second ingredient in the specification of a dynamics are the {\it mathematical objects} that  change with time due to the dynamics. In order to specify this concept, we start introducing a general notion of {\it configuration space} suitable for our purposes.

It is required that the configuration space $\mathcal{M}$ is equipped with a topological structure. The existence of a topology in $\mathcal{M}$ allows us to consider continuous dynamical laws. Continuity of the dynamical law is a useful condition to establish cause-effect descriptions of dynamical processes. Therefore, the topology on $\mathcal{M}$ must be consistent with all natural causal-effect relations admissible in the theory. Let us assume that a given quantity $E(t)$ does not evolve continuously on time. In the case it is difficult to follow the details of an evolution, one could instead consider differences on measurements $E_2-\,E_1$, where the detailed time dependence has been erased. Then a large change of the form $E_2-E_1$ can be associated either to a short time evolution $t_2/t_1 \approx 1$ through a local causal explanation, or it can associated to different values at different times $t_2/t_1>>1$  through a global causal explanation. Therefore, with a non-continuous law, that does not restrict the amount of change due to the small changes of the time parameter $t$, it is more difficult to discriminate a local causation from a global causation of a large change of the form $E_2-E_1$.
The absence of continuity does not imply a contradiction in the theory, but the identification of a global cause $t_1\to t_2$ with $t_2>>t_1$ of a large change $E_1\to E_2$ with $E_2>> E_1$ is more laborious problem than the identification of a local causes of small changes.

From the above argument, we will consider a theory of dynamics in the category of topological spaces and topological maps and  where the dynamics will be continuous maps.  Hence we adopt the following notion of configuration space,
\begin{definicion}
The {\it configuration space} $\mathcal{M}$  of a dynamical system is a topological space whose elements describe the state of the system.
\label{configuration space}
\end{definicion}
The configuration space $\mathcal{M}$ can be either a discrete set or a continuous set. Also, the time parameters can  be discrete or continuous. Respect to the topology, every set ${\bf T}$ can be endowed always with at least two topologies, namely the discrete and the indiscrete topology \cite{Kelley}. Between these two extreme topologies, there can be many others. In our case, a natural topology  on the configuration space $\mathcal{M}$ is the minimal topology where the Hausdorff separability condition holds. Obviously, for the discrete topology, the Hausdorff condition holds. But the discrete topology could be too restrictive for a particular dynamics.

Our notion of configuration space applies to both classical dynamical systems, where $\mathcal{M}$ is a classical configuration space, but also can be applied to quantum dynamical systems. For a quantum system, $\mathcal{M}$ is a projective Hilbert space $\mathcal{H}$ and the state of the system is described by elements of  $\mathcal{H}$.
A dynamical system will have associated a sub-domain of elements of $\mathcal{M}$ during the time evolution. Moreover, the dynamical variables can also be fields defined over the configuration space $\mathcal{M}$.

The type of mathematical structures over $\mathcal{M}$ that will be considered must allow the formulation of continuous laws for dynamics. For this purpose  we apply concepts from sheave theory  \cite{Godement}, \cite{Hirzebruch1978}, $\S$ 2. The notion of sheaf offers the natural setting to speak of fields over $\mathcal{E}$ as sections. Let us consider a $\mathcal{A}$-sheaf $(\mathcal{E},\pi_\mathcal{E},\mathcal{M})$, where $\pi_\mathcal{E}:\mathcal{E}\to \mathcal{M}$ is continuous. The algebraic structures of the stalks $\mathcal{A}_u =\,\pi^{-1}(u)$ are such that the corresponding composition operations are continuous on each stalk.
 Typical algebraic structures for the stalks $\mathcal{A}_u$ to be considered are $\mathbb{K}$-vector fields of finite dimension and algebraic geometric constructions.
  A section of a sheaf is a continuous map $E:\mathcal{M}\to \,\mathcal{E}$ such that $\pi_\mathcal{E}\circ E =\,Id_\mathcal{M}$. Then we propose the following notion of field,
 \begin{definicion}
 A field $E$ is a section of a sheaf $\pi_\mathcal{E}:\mathcal{E}\to \mathcal{M}$.
 \end{definicion}
 For an abelian sheaf, where the stalks are abelian groups, the zero section $x\mapsto 0_u\in \,\mathcal{A}_u$ is a continuous section. Any small continuous deformation of the zero section is also a continuous section.

 The set of sections of a sheaf $\mathcal{E}$ is denoted by $\Gamma\,\mathcal{E}$. The algebraic operations on the stalks induce operations on sections of $\Gamma \mathcal{E}$.

\subsection{Notion of dynamical law} The third ingredient that we need in the formulation of a general theory of dynamics is a notion of dynamical law compatible with the above notions of time parameters and configuration spaces.
\begin{definicion}
Given a configuration space $\mathcal{M}$, a number field  $(\mathbb{K},+,\cdot)$ and a time parameter $J\subset \mathbb{K}$,
a {\it local dynamics} or {\it flow} is a continuous map in the product topology of $J\times \,\mathcal{M}$,
\begin{align*}
\Phi:J\times \,\mathcal{M}\to \mathcal{M},\quad (t,u)\mapsto \Phi_t(u)
\end{align*}
such that
\begin{itemize}
\item The following group composition condition holds:
\begin{align}
\Phi(t_1,u)\circ \Phi(t_2,u)=\,\Phi(t_1+t_2,u),\quad
\label{composition law for Phi}
\end{align}
whenever both sides are defined.
\item
The condition
\begin{align}
\Phi_0(u)=\,u
\end{align}
 holds for every $u\in\,\mathcal{M}$, where $0$ is the neutral element of the sum operation $+:\mathbb{K}\times\mathbb{K}\to \mathbb{K}$.

\end{itemize}
\label{definitionofdynamics}
 \end{definicion}
 The relation \eqref{composition law for Phi} is very often  re-written in dynamical theory in the form
  $\Phi_{t_1+t_2}(x)=\,\Phi_{t_1}\circ \Phi_{t_2}$, where $\Phi_{t}=\Phi(t,\cdot)$.
 One can compare this notion of dynamics with standard notions of dynamical systems, for instance as discussed in  \cite{ArnoldAvez} or as discussed in \cite{Sternberg}, chapter 12.

 The term {\it local dynamics} refers to the fact that the outcome of the evolution depends pointwise on $\mathcal{M}$. Locality can be dropped from the notion of dynamics, but theory of non-local dynamics will be more involved.

 {\it Definition} \ref{definitionofdynamics} does not require that the time parameter set $J$ must be a subset of a number field, being only required that the addition operation $+:J\times J\to J$ is defined. However, in order to define incremental quotients as given by the expressions of the form \eqref{incremental quotient}  of fields defined over $\mathcal{M}$, $J$  must be a subset of a number field $\mathbb{K}$ to secure that the expressions of the form $ (\delta t)^{-1}=(t_2-t_1)^{-1}$ are defined for $t_2-t_1$ for sufficiently small but non-zero elements.

 If there is defined an order relation in the number field $\mathbb{K}$, then there is an induced order relation in $J$. In this case, time ordered sequences from $A\in \mathcal{M}$ to $B\in\mathcal{M}$ along the dynamics $\Phi$ can be defined and a chronological order can be attached to the evolution from $A$ to $B$ by means of $\Phi$.
  If the number field $\mathbb{K}$ is not ordered, then there is no notion of local time ordering for the dynamics $\Phi$. In such a case, there is no notion of global time ordering as it usually appears in relativistic spacetime models. From the general point of view discussed until here, the number field $\mathbb{K}$ that appears in the notion of dynamics does not need to be ordered, but the number field $\mathbb{K}'$ where the distance function $d_\mathbb{K}:\mathbb{K}\times\,\mathbb{K}\to \,\mathbb{K}'$ takes values, must be an ordered number field, in order to accommodate the existence of a quasi-metric function $d_{\mathbb{K}}$ in $\mathbb{K}$.

 {\it Definition} \ref{definitionofdynamics} is an straightforward generalization of the usual notion of dynamical system \cite{ArnoldAvez}. However, in {\it definition} \ref{definitionofdynamics} it is not required any explicit condition on a measure on the configuration space $\mathcal{M}$.
\begin{definicion}
The local dynamics $\Phi:J\times \mathcal{M}\to \mathcal{M}$ is complete if the time parameter set $J$ coincides with $\mathbb{K}$ as a set.
\label{complete dynamics}
\end{definicion}
 \begin{proposicion}
 Let $\Phi$ be a complete local dynamics.
The transformations $\{\Phi_t\}_{t\in\,J}$ according to {\it definition} \ref{definitionofdynamics} determine a group of transformations of $\mathcal{M}$.
\label{group of transformations}
\end{proposicion}

 The extension of the notion of local dynamics to the evolution of mathematical objects defined over the configuration space $\mathcal{M}$ is done in the following way. Let us consider a $\mathbb{K}$-module sheaf $\pi_{\mathcal{E}}:\mathcal{E}\to \mathcal{M}$ and a dynamics $\Phi:J\times \mathcal{M}\to \mathcal{M}$. Let $\varphi:\mathbb{K}\to \mathbb{K}$ be a  $\mathbb{K}$-isomorphism,
 \begin{align*}
 &\varphi(t_1+t_2)=\,\varphi(t_1)+\,\varphi(t_2),\\
 & \varphi(t_1 \cdot \,t_2)=\,\varphi(t_1)\cdot \,\varphi(t_2).
 \end{align*}
 The simplest case is to consider the identity map $\varphi= Id_\mathbb{K}$, but this restriction is not strictly necessary.
A continuous map $\Phi_{\mathcal{E}}:\mathbb{K}\times  \mathcal{E}\to \mathcal{E}$ is such that the diagram
 \begin{align}
\xymatrix{\mathbb{K}\times \mathcal{M} \ar[r]^{\Phi}   & \mathcal{M}\\
\mathbb{K}\times \mathcal{E} \ar[r]^{\Phi_{\mathcal{E}}} \ar[u]^{\varphi\times \pi_{\mathcal{E}}}  & \mathcal{E} \ar[u]_{\pi_{\mathcal{E}}} }
\label{commutativity of the diagram on E}
\end{align}
commutes.
 \begin{proposicion} Let $\Phi$ be a local dynamics on $\mathcal{M}$ and consider the induced continuous map $\Phi_\mathcal{E}$ on the sheaf $\mathcal{E}$ such that the diagram \eqref{commutativity of the diagram on E} commutes. Then for any section $E\in\,\Gamma \mathcal{E}$ there is an open  neighborhood $N\subset \mathcal{E}$ such that
\begin{align}
\Phi_\mathcal{E}(t_1+t_2,\cdot)=\,\Phi_\mathcal{E}(\varphi(t_2),\Phi_\mathcal{E}(\varphi(t_1),\cdot))
\label{composition Phi for E}
\end{align}
holds good on $N$, whenever both sides of the relation \eqref{composition Phi for E} are well defined.
\label{proposicion on law group for the dynamics}
\end{proposicion}
\begin{proof}
   The commutativity of the diagram \eqref{commutativity of the diagram on E} and the homomorphism law of the dynamics for $\Phi$, eq. \eqref{composition law for Phi} imply that
\begin{align*}
\pi_\mathcal{E}\circ \Phi_\mathcal{E}(t_1+t_2,e_u)& =\,\Phi(\varphi(t_1+t_2),u)\\
& =\,\Phi(\varphi(t_1)+\varphi(t_2),u)\\
& = \Phi(\varphi(t_1),\Phi(\varphi(t_2),u))\\
& =\,\Phi(\varphi(t_1),\Phi(\varphi(t_2),\pi_\mathcal{E}(e_u)))\\
& =\,\Phi(\varphi(t_1),\pi_\mathcal{E}\circ \Phi_\mathcal{E}(t_2,e_u))\\
& =\,\pi_\mathcal{E}\circ \Phi_\mathcal{E}(t_1,\Phi_\mathcal{E}(t_2,e_u)),
\end{align*}
for every element $e_u$ of the stalk $E_u$.
Since the restriction of  $\pi_\mathcal{E}$ in some open neighborhood $N\in \mathcal{E}$ is an homeomorphism, then it follows the relation \eqref{composition Phi for E}.
\end{proof}
 A continuous map $\Phi_{\mathcal{E}}:\mathbb{K}\times \,\mathcal{E}\to\mathcal{E}$ is and induced dynamics in $\mathcal{E}$ from $\Phi$. Given the sheaf $\pi_\mathcal{E}:\mathcal{E}\to \mathcal{M}$, the dynamics $\Phi_\mathcal{E}$ is not necessarily a morphism and it does not need to be unique. Even when $\pi_\mathcal{E}:\mathcal{E}\to \mathcal{M}$ is a $\mathbb{K}$-module sheaf, it does not imply $\Phi_\mathcal{E}$  that is linear.

An analogous notion can be formulated for the associated dynamics $\Phi_\mathcal{E}$ acting on sections of the sheaf $\pi_\mathcal{E}:\mathcal{E}\to \mathcal{M}$. It is determined by the commutativity of the diagram 
 \begin{align}
\xymatrix{\mathbb{K}\times \mathcal{M}\ar[d]_{\varphi^{-1}\times E} \ar[r]^{\Phi}   & \mathcal{M}\ar[d]^{E} \\
\mathbb{K}\times \mathcal{E} \ar[r]^{\Phi_{\mathcal{E}}}   & \mathcal{E}  }
\label{commutativity of the diagram on E 2}
\end{align}
for every section $E\in\,\Gamma \mathcal{E}$.

\subsection{On time re-parametrization invariance}
Our notion of dynamics assumes the existence of time parameters as subsets $J\subset\,\mathbb{K}$ furnished with several consistent composition and other conditions. For a general dynamics, time parameters can be taken quite arbitrarily. As a consequence, time parameters will lack an observational or phenomenological interpretation. Since there is no macroscopic observer attached to such parameters, it is natural to expect that in a consistent description of a physical dynamics, physical quantities must be independent of the choice of the time parameter for the fundamental dynamics. By this we mean that physical quantities are equivalent classes of mathematical objects $[\Upsilon_0]$ which are covariantly defined: for every time re-parametrization $\theta:J\to \mathbb{K}$ there is at least two representatives $\Upsilon,\widetilde{\Upsilon}\in\,[\psi_0]$ such that $\widetilde{\Upsilon}(\theta (t))=\,\Upsilon(t)$ for $t\in\,J$. Therefore, although for the description of the dynamics the introduction of time parameters is necessary, the dynamics that are relevant for fundamental physical applications must be time re-parametrization invariance.

\section{Notions of local reversible and local non-reversible dynamics}
\begin{definicion}
 The {\it time conjugated dynamics} associated to the dynamics $\Phi:J\times \,\mathcal{M}\to \mathcal{M}$ with $J\subset \mathbb{K}$ is a map
 \begin{align*}
 \Phi^c:J\times \,\mathcal{M}\to \mathcal{M}
 \end{align*}
  such that if $\Phi(t,A)=B$, then it must hold that $\Phi^c(t,B)=A$
 for every $A,B\in\, \mathcal{M}$.
 \label{conjugate dynamics}
\end{definicion}
Since $t_1+t_2=\,t_2+t_1$ for any pair of elements $t_1,t_2\in\mathbb{K}$, it follows that if $\Phi$ is a dynamics, then $\Phi^c $ is also a dynamics and both are defined using the same time parameter $J\subset\,\mathbb{K}$.
It is direct that the idempotent property holds good,
\begin{align}
\left(\Phi^c\right)^c=\,\Phi
\label{idempotent property of the conjugate}.
\end{align}

\begin{definicion}
Let  $\Phi_{\mathcal{E}}:J\times \Gamma\mathcal{E}\to\,\Gamma\mathcal{E}$ with $J\subset\,\mathbb{K}$ be a dynamics. The conjugate dynamics is a map
$\Phi^c_{\mathcal{E}}:J\times \Gamma\mathcal{E}\to\,\Gamma\mathcal{E}$ such that if $\Phi_\mathcal{E}(t,E_1)=E_2$, then $\Phi^c_{\mathcal{E}}(t,E_2)=\,E_1$ for every $E_1,E_2\in\,\Gamma\,\mathcal{E}$, where $\mathcal{E}$ is a $\mathbb{K}$-module sheaf $\pi_{\mathcal{E}}:\mathcal{E}\to \,\mathcal{M}$.
\end{definicion}

For the conjugate dynamics $\Phi^c_\mathcal{E}$, the idempotent property
\begin{align}
\left(\Phi^c_\mathcal{E}\right)^c=\,\Phi_\mathcal{E}
\label{idempotent property of the conjugate 2}
\end{align}
holds good.

In the category of topological spaces and continuous functions, the first natural criteria to decide if a given dynamics $\Phi$ is reversible is to assume the condition
\begin{align*}
\lim_{t\to 0}\Phi(t,E)=\,\lim_{t\to\,0} \Phi^c(t,E)
\end{align*}
for all $\Gamma \,E\in\,\mathcal{E}$.
But this condition is always satisfied in the category of topological spaces and topological maps, if $\Phi$ (and hence $\Phi^c$) are dynamics, since
\begin{align*}
\lim_{t\to\,0}\Phi(t,E)=\,\lim_{t\to\,0} \Phi^c(t,E)=E.
\end{align*}
Instead, we propose a notion of local irreversibility based upon the following
 \begin{definicion}
 Let ${\Phi}:J\times \mathcal{M}\to \,\mathcal{M}$ be a dynamics over $\mathcal{M}$. The dynamics $\Phi$ is non-reversible if there is a sheaf $\pi_\mathcal{E}:\mathcal{E}\to \mathcal{M}$ and a continuous function
 $\Omega: \Gamma\,\mathcal{E}\to \mathbb{K}$  such that for the induced dynamics $\Phi_{\mathcal{E}}:J\times\Gamma \mathcal{E}\to \Gamma\mathcal{E}$, the relation
 \begin{align}
 \Xi_\Omega:\Gamma\mathcal{E}\to \mathbb{K},\,E\mapsto\lim_{t\to\,d_{\mathbb{K}min}}\,\frac{1}{t}\,\left(\Omega\circ\Phi_\mathcal{E}(t,E)-\,\Omega\circ \Phi^c_\mathcal{E}(t,E)\right)\neq 0
 \label{definitionofnonreversibilityfunction}
 \end{align}
 holds good for all $E \in\, \Gamma\,\mathcal{E}$, whenever $t_1,t_2,t_1+t_2 \in J\subset \mathbb{K}$.

 A dynamics which is not non-reversible in the above sense is called reversible dynamics.

 A function $\Xi$ for which the condition \eqref{definitionofnonreversibilityfunction} holds will be called a {\it non-reversibility function}.
 \label{nonreversibledynamics}
 \end{definicion}
When $\Gamma\,\mathcal{E}$ is equipped with a measure, the non-reversibility condition  \eqref{definitionofnonreversibilityfunction} can be formulated for almost all $E\in\,\Gamma\,\mathcal{E}$, that is, for all subsects in $\Gamma \,E$ except maybe for subsets of measure zero in $\Gamma\,\mathcal{E}$.

If instead of the expression  \eqref{definitionofnonreversibilityfunction}, one considers as a criteria for non-reversibility the conditions
 \begin{align*}
\Delta:\Gamma\mathcal{E}\to \mathbb{K},\,E\mapsto &\lim_{t\to\,0}\,\left(\Omega\circ \Phi_\mathcal{E}(t,E)-\,\Omega\circ \Phi^c_\mathcal{E}(t,E)\right)\\
& =\, \lim_{t\to\,0}\,\left(\Omega\circ\Phi_\mathcal{E}(0,E)-\,\Omega\circ \Phi^c_\mathcal{E}(0,E)\right)\\
& =\, \Omega(E)-\,\Omega(E)=\,0,
 \end{align*}
since $\Phi_\mathcal{E}(0,E)=\,\Phi^c_\mathcal{E}(0,E)=\,E$. This result holds for any continuous dynamics $\Phi$. But the category of topological spaces with continuous functions as a maps is the natural category where to formulate our mathematical models. By the arguments discussed above, continuity is an essential ingredient for determinism, and for the construction of deterministic models. Hence this alternative definition does not allow to define a notion of local non-reversible law in the category of topological spaces.

 In order to apply the above notion of non-reversibility, it is  requirement that the number field $\mathbb{K}$ and the configuration space $\mathcal{M}$ admit a well defined notion of the limit operation $t\to\,d_{\mathbb{K}min}$ as it appears in the expression \eqref{definitionofnonreversibilityfunction} and also as it appears in the notion of incremental quotient limit, given by the expression \eqref{incremental quotient}. For example, the field of real numbers $\mathbb{R}$ and the discrete field of rational numbers $\mathbb{Q}$ have well defined notions of  the limit $t\to\,0$. Other examples are constituted by the field of complex numbers $\mathbb{C}$ and the prime field ${\mathbb Z_p}$. In such a case, the notion of limit $t\to d_{\mathbb{K}min}=1$ must be well defined.

 If a dynamics $\Phi$ is non-reversible, then it must exists  sheaf where a function $\Omega:\Gamma \,\mathcal{E}\to \mathbb{K}$ such that the non-reversibility function $\Xi_\Omega$ is different from zero. Conversely, a reversible dynamics $\Phi:J\times \mathcal{M}\to \mathcal{M}$ is such that for any sheaf $\mathcal{E}$ and any continuous function $\Omega:\mathcal{E}\to \mathbb{K}$, $\Xi_\Omega =0$. Therefore, it is  theoretically easier to check if a given dynamics is non-reversible than to check if it is reversible, because for the first option it is enough to find a function  $\Omega:\Omega: \Gamma\,\mathcal{E}\to \mathbb{K}$ such that the relation \eqref{definitionofnonreversibilityfunction} is satisfied, while for the second option  one needs to check that for all such functions $\Omega$, the non-reversibility $\Xi_\Omega$ is identically zero.

 Let us consider a non-reversible dynamics $\Phi:J\times \mathcal{M}\to \mathcal{M}$ such that for an associated dynamics $\Phi_\mathcal{E}$, there is a non-zero reversibility function $\Xi_\Omega\neq 0$. Because of the algebraic structures of the stalk $E_u$, it is possible to define the map
 \begin{align}
 Sym \Phi_\mathcal{E}:J\times \Gamma\mathcal{E}\to \,\Gamma \mathcal{E},\quad (t,E)\mapsto \,\frac{1}{2}\,\left(\Phi_\mathcal{E}(t,E)+\,\Phi^c_\mathcal{E}(t,E)\right).
 \label{symmetrization of the dynamics}
 \end{align}
For $Sym \Phi_\mathcal{E}$ the property \eqref{composition Phi for E} holds good. Furthermore, $Sym \Phi_\mathcal{E}$ one has that $\Xi_\Omega =0$, indicating that the operation of symmetrization in \eqref{symmetrization of the dynamics} is a form of reducing non-reversible to reversible dynamics. Repeating this procedure for any induced dynamics $\Phi_\mathcal{E}$ on each sheaf $\pi_\mathcal{E}:\mathcal{E}\to \mathcal{M}$, we can assume the existence of an induced dynamics $Sym \Phi:J\times \mathcal{M}\to \mathcal{M}$ that by construction will be reversible. When such dynamics  exists, $Sym \Phi$ is called the {\it symmetrized dynamics} on $\mathcal{M}$.

Let us consider conditions under which $\Xi_\Omega\equiv 0$. The smoothness of  $\Omega:\Gamma\mathcal{E}\to \mathbb{K}$ is expressed through formal Taylor's developments,
\begin{align*}
&\Omega\circ\Phi_\mathcal{E}(t,E)=\,\Omega \circ \Phi_\mathcal{E}(0,E)+\,t\,\Omega'\star d\Phi_\mathcal{E}(0,E)+\,\mathcal{O}(t^2),\\
&\Omega\circ \Phi^c_\mathcal{E}(t,E)=\,\Omega\circ\Phi^c_\mathcal{E}(0,E))+\,t\,\Omega'\star d\Phi^c_\mathcal{E}(0,E)+\,\mathcal{O}(t^2),\\
\end{align*}
where the differentials is of the form
\begin{align*}
\Omega'\star d\Phi_\mathcal{E}(0,E) & :=\,\frac{d{\Omega}(u)}{d u}\large|_{u=\Phi_\mathcal{E}(0,E)}\star \,\frac{d\Phi_{\mathcal{E}}(t,E)}{dt}\large|_{t=0}\\
& =\,\frac{d{\Omega}(u)}{d u}\large|_{u=E}\star \,\frac{d\Phi_{\mathcal{E}}(t,E)}{dt}\large|_{t=0},
\end{align*}
and
\begin{align*}
\Omega'\star d\Phi^c_\mathcal{E}(0,E) & :=\,\frac{d{\Omega}(u)}{d u}\large|_{u=\Phi^c_\mathcal{E}(0,E)}\star \,\frac{d\Phi^c_{\mathcal{E}}(t,E)}{dt}\large|_{t=0}\\
& =\,\frac{d{\Omega}(u)}{d u}\large|_{u=E}\star \,\frac{d\Phi^c_{\mathcal{E}}(t,E)}{dt}\large|_{t=0}.
\end{align*}
The $\star$-pairing is the natural pairing induced from $\Omega$ operating on $\Phi_\mathcal{E}(t,E)$ by the natural composition $\Omega\circ \,\Phi_\mathcal{E}$.
Therefore,  $\Xi_\Omega$ can be written formally as
\begin{align*}
\Xi_\Omega(E)& =\,\lim_{t\to\,0}\,\frac{1}{t}\,\left(t\,\Omega'\star \,d\Phi_\mathcal{E}(0,E)-\,t\Omega'\star \,d\Phi_\mathcal{E}(0,E))\right)\\
& =\,\Omega'\star\,d\Phi_\mathcal{E}(0,E)-\,\Omega'\star\,d\Phi^c_{\mathcal{E}}(0,E).
\end{align*}
Taking into account the above expressions, we have
\begin{align*}
\Xi_\Omega(E) =\frac{d{\Omega}(u)}{d u}\large|_{u=E}\star \,\frac{d\Phi_{\mathcal{E}}(t,E)}{dt}\large|_{t=0}
-\,\frac{d{\Omega}(u)}{d u}\large|_{u=E}\star \,\frac{d\Phi^c_{\mathcal{E}}(t,E)}{dt}\large|_{t=0},
\end{align*}
which is in principle, different from zero. This shows that the criteria to decide when a dynamics is non-reversible is well-defined in the category of topological spaces. 
\\
{\bf General form of the reversibility condition}. The non-reversibility function $\Xi$ is identically zero if and only if the condition 
\begin{align}
\Omega'\,\star \left(\,\frac{d\Phi_{\mathcal{E}}(t,E)}{dt}\large|_{t=0}\,- \frac{d\Phi^c_{\mathcal{E}}(t,E)}{dt}\large|_{t=0}\right) =0
\label{characterization of reversible dynamics}
\end{align}
holds good for every continuous function $\Omega: \Gamma\,\mathcal{E}\to \mathbb{K}$. If the $\star$-pairing is in appropriate sense invertible, the relation \eqref{characterization of reversible dynamics} can be re-written formally as a necessary condition that depends only on the dynamical law,
\begin{align}
\frac{d\Phi_\mathcal{E}(t,E)}{dt}\Big|_{t=0}-\, \frac{d\Phi^c_\mathcal{E}(t,E)}{dt}\Big|_{t=0}\equiv \,0.
\label{characterization of reversible dynamics 2}
\end{align}
Either the condition \eqref{characterization of reversible dynamics} or the condition \eqref{characterization of reversible dynamics 2} can be taken as the necessary and sufficient condition for reversibility of a local dynamics.
\bigskip
\\
{\bf Notion of non-reversible dynamics in configuration spaces endowed with a measure}. For spaces endowed with a measure, the criteria for non-reversibility is the existence of a non-reversibility function $\Xi_\Omega:\Gamma\mathcal{E}\to \mathbb{K}$ which is  non-zero except in a subset of measure zero. The relevant point is that,  for a given measure on $\mathcal{E}$, $\Xi_\Omega$ is non-zero almost everywhere during the evolution. If there is no such a function $\Omega:\Gamma\mathcal{E}\to \mathbb{K}$ for a dense subset in an open domain $\mathcal{U}_A\subset \mathcal{M}$ containing $A$, then one can say that the dynamics is reversible locally.
However,  strict non-reversible laws or strict reversible laws, namely dynamical laws which are non-reversible (resp. reversible) in the whole configuration space $\mathcal{E}$, avoids the introduction of a measure in $\mathcal{E}$ as we did in our definition \ref{nonreversibledynamics}.

\section{The relation between non-reversible dynamics and a generalization of the second principle of thermodynamics}
In order to introduce thermodynamic systems, let us consider the following product of topological spaces, 
\begin{align}
\widetilde{\mathcal{M}}=\,\prod^N_{k=1}\,\mathcal{M}_k,
\label{thermodynamic product space}
\end{align}
where each of the spaces $\mathcal{M}_k$ is by assumption the configuration space of a given dynamical system and $N$ is a large natural number. By a 
large integer $N$ we mean the following asymptotic characterization: there is a collection of maps $P[N]:\widetilde{\mathcal{M}}\to \mathbb{K}$ that depends upon $N=\,\dim(\mathcal{M})$ and such that 
\begin{align}
 P[N]=\,P[N-1]+{o}(N^\delta),
 \label{asymptotic property}
 \end{align}
  with $\delta <\,0$. That is, we assume the condition
\begin{align*}
\lim_{N\to \,+\infty} \frac{P[N]-P[N-1]}{N^\delta}\to 0.
\end{align*}
A topological space $\widetilde{\mathcal{M}}$ with this asymptotic property will be called a {\it thermodynamic space}, where a topological space since it embraces the interpretation of thermodynamic systems as the ones where it is possible to define local intensive and extensive functions of the whole system where fluctuation effects due to the detailed structure of the system can, with great accuracy, be disregarded (see for instance \cite{Kondepudi Prigogine 2015}, chapter 15). $P:\widetilde{\mathcal{M}}\to \mathbb{K}$ is called thermodynamic function. For this class spaces one can speak of thermodynamic sub-systems if $\mathcal{M}'_\hookrightarrow \,\widetilde{\mathcal{M}}$ are embedding for which the asymptotic conditions \eqref{asymptotic property} holds. Furthermore, since $\widetilde{\mathcal{M}}$ is by definition a large product space, our notion of thermodynamic space provides the setting for statistical interpretations of the maps ${P}:\widetilde{\mathcal{M}}\to \mathbb{K}$.
\begin{ejemplo}
Let us consider $\mathcal{M}=\prod^N M_k$ with each configuration space of the form $M_k \cong M$. Then each of the spaces ${M}_k$ is not a thermodynamic space, since if $\mathcal{M}\cong {M}_k$, then $N=1$ and the above asymptotic property does not hold. This is in agreement with the idea that a system composed by a single individual element that can be characterized within a given theory as fundamental, is not a thermodynamic system.
\end{ejemplo}

Let $\widetilde{\mathcal{M}}_c$ be a classical phase space and $J\subset \mathbb{R}$ open.
The notion of entropy in  classical equilibrium thermodynamics corresponds to a map of the form
\begin{align}
\Lambda_c:J\times \widetilde{\mathcal{M}}_c\to \mathbb{R},
\label{entropy function}
\end{align}
 such that
 \begin{enumerate}
 \item The function $\Lambda_c$ is extensive: for two thermodynamicly different classical thermodynamic spaces $\widetilde{\mathcal{M}}_{1c}\hookrightarrow \mathcal{M}_c$, $\widetilde{\mathcal{M}}_{2c}\hookrightarrow \,\mathcal{M}_c$ corresponding to two sub-systems of the thermodynamic system $\mathcal{M}_c$ with entropy functions  $\Lambda_{1c}:J\times \widetilde{\mathcal{M}}_{1c}\to \mathbb{R}$ and $\Lambda_{2c}:J\times \widetilde{\mathcal{M}}_{2c}\to \mathbb{R}$,  then it must follow that
 \begin{align*}
 \Lambda_c\left(t,(u_1,u_2)\right)\geq \,\Lambda_{1c}\left(t,u_1\right)+\,\Lambda_{2c}\left(t,u_2\right).
 \end{align*}

 \item For any thermodynamic system, it is non-decreasing with time,
 \begin{align*}
\frac{d}{d t} \Lambda_c\left(t,u\right)\geq 0.
 \end{align*}

 \item For any thermodynamic subsystem described by $\widetilde{\mathcal{N}}_c$ subset of $\widetilde{\mathcal{M}}_c$, it is non-decreasing with time.

  \end{enumerate}
  This is not a full characterization of the entropy function, but it will serve for our purposes and it is consistent with the properties of the Boltzmann entropy \cite{Tolman}.

Let us consider now a generalized form of the second principle of thermodynamics for general dynamical systems.
  We first extend the above  notion of entropy function by considering entropy functions as extensive maps of the form
  \begin{align*}
  \Lambda:J\times \widetilde{\mathcal{M}}\to \mathbb{K},
  \end{align*}
 where $\mathbb{K}$ is an ordered number field and $\widetilde{\mathcal{M}}$ is the product space of the form \eqref{thermodynamic product space} and such that the above properties $(1)-(3)$ of the classical entropy hold good for the function $\Lambda$. Also, this characterization can be applied to  {\it local entropy densities} $\Lambda_i$. The notion of local internal entropy density appears in the theory of linear non-equilibrium thermodynamics and are described for instance in \cite{Kondepudi Prigogine 2015}, chapter 15.
In this general context of thermodynamic systems and generalized form of entropy functions, the second principle of thermodynamics can be stated as follows:
\bigskip
\\
{\bf Generalized Second Principle of Thermodynamics}: {\it The dynamical change in the state describing the evolution of a thermodynamic system is such that the production of internal local density entropy function $\Lambda_i(u)$ for each of the $i=1,...,r$ thermodynamic sub-systems do not decrease during the time evolution.}
\\

In order to compare the  notion of local non-reversible dynamical law as understood in definition  \ref{nonreversibledynamics} with the notion of physical evolution according to the non-decrease of entropy of the second principle of thermodynamics, the first step is to discuss the relations between the relevant notions, namely the non-reversibility function \eqref{definitionofnonreversibilityfunction} and the notion of generalized entropy function. Let us consider the section $E\in\,\Gamma \mathcal{E}$ of the sheaf $\pi_\mathcal{E}:\mathcal{E}\to \mathbb{K}$, $\Phi_\mathcal{E}:J\times \Gamma \mathcal{E}\to \mathcal{E}$ the induced local dynamics and $\Omega:\Gamma \mathcal{E}\to \mathbb{K}$ a scalar function such that $\Xi_\Omega$ is locally positive. Then the correspondence between the generalized second principle and the notion of non-reversible law is such that
\begin{align*}
\frac{d}{dt}\large|_{t=0}\Lambda (t,u)\equiv \,\frac{d}{dt}\large|_{t=0}\left(\Omega\,\circ \Phi_\mathcal{E}(t, E (u))\,-\Omega\,\circ \Phi^c_\mathcal{E}(t,\, E (u))\right)=\Xi_\Omega(E(u)).
\end{align*}
In principle, this relation depends on the section $E:\mathcal{M}\to \mathcal{E}$ and the scalar map $\Omega$, revealing a different entropy function for each choice of the pair $(E,\Omega)$. For instance, the notion of local thermodynamic equilibrium states, given by the extremal condition  $\frac{d}{dt}\large|_{t=0}\Lambda (t,u)=0$, for a parameter $t\in J_0$ such that the equilibrium is found for $t=0$. But such equilibrium condition is not consistent with the dependence of $\Xi_\Omega$ on the section $E$ and also on the map $\Omega$. However, if the section $E$ is fixed by a physical principle, the relation between the formalism depends upon the choice of the function $\Omega$ only up to  a constant. As we will see below, this will be the case of quantum dynamics.

According to the above discussion, there must be two different notions of irreversibility. The first one is related to the notion of non-reversible dynamics according to definition \ref{nonreversibledynamics}. The second one is associated with the generalized second principle of thermodynamics and is applied to thermodynamic systems.
The  distinction between these two notions of irreversibility is important, since the notion of non-reversible dynamics based upon definition \ref{nonreversibledynamics} can be applied to a general dynamics. In particular, it can  be applied to fundamental systems described by points of the configuration space $\mathcal{M}$, where thermodynamic relations between thermodynamic variables is not justified. Despite this general situation, the two notions of non-reversibility can be identified in relevant examples, as we will discuss later.

\section{Examples of reversible and non-reversible dynamics}
 An interesting and relevant example of non-reversible dynamics is the following,
 \begin{ejemplo}
 Let $(M,F)$ be a Finsler space   \cite{BaoChernShen}, where $M$ is an $m$-dimensional manifold and $\mathbb{K}$ is the field of real numbers $\mathbb{R}$. For a point $A$ and for a variable point $X$ infinitesimally close to $A$, one considers the limits
 \begin{align*}
  &\lim_{s\to 0}\,\frac{1}{s}\,\Omega(\Phi_t(X))=\,\frac{d}{ds}\big|_{s=0}\,\int^{X(s)}_A\,F^2(\gamma,\dot{\gamma})\,dt,\quad\\
   &\lim_{s\to 0}\,\frac{1}{s}\,\Omega(\Phi^c_t(X))=\,\frac{d}{ds}\big|_{s=0}\,\int^A_{X(s)} \,F^2(\tilde{\gamma},-\dot{\tilde{\gamma}})\,dt.
 \end{align*}
  The dynamics $\Phi$ is given by the geodesic flow of $F$. The parameterized curves $\gamma$ and $\tilde{\gamma}$ are geodesics,
 where $\gamma$ is a curve joining $A$ and $X$ and realizing the minimal length for curves joining $X$ and $A$ (and analogously for the inverted geodesic $\tilde{\gamma}$).  The local existence of such geodesics is guaranteed by Whitehead theorem \cite{Whitehead1932}. However, these pair of geodesics are not related as they are in Riemannian geometry by a relation of the form $\tilde{\gamma}(s)=\,\gamma(1-s),\,s\in [0,1]$, since in general the Finsler metric $F$ is not necessarily a {\it reversible Finsler metric}.

  Let us  consider the expression
  \begin{align*}
 \Xi_\Omega(A) & =\,\lim_{s\to 0}\,\frac{1}{s}\,\Omega(\Phi_t(X))-\,\lim_{s\to 0}\,\frac{1}{s}\,\Omega(\Phi^c_t(X)) \\
 & = \,\frac{d}{ds}\big|_{s=0}\,\left\{ \int^{X(s)}_A\,\left(F^2(\gamma,\dot{\gamma})-\, F^2(\tilde{\gamma},-\dot{\tilde{\gamma}})\right)\,dt\right\}.
 \end{align*}
 It follows that
 \begin{align}
 \lim_{X\to A}\, \Omega(X)-\Omega^c(X^c)=\,F^2(A,V)-\, F^2(A,-V)
 \label{FmenosF}
 \end{align}
 along an integral geodesic curve $\gamma:I\to M$ passing through $A$ and $B$ and such that
 \begin{align*}
 A=\,\gamma(0)=\,\lim_{X\to A}\,\tilde{\gamma}, \quad V=\,\dot{\gamma}(0)=\,-\lim_{X\to A}\,\dot{\tilde{\gamma}}.
 \end{align*}
 For a generic Finsler metric, the limit \eqref{FmenosF} is different from zero. This is also true for geodesics and since the geodesics are determined locally by the initial conditions, via Whitehead's theorem \cite{Whitehead1932}, we have that the geodesic dynamics $\Phi$ is in this case non-reversible with
 \begin{align}
  \Xi_\Omega(A)=\,F^2(A,V)-\, F^2(A,-V),
 \end{align}
which is non-zero for almost all $(A,V)\in TM$.

This example is related with the notion of {\it reversibility function} in Finsler geometry \cite{Rademacher}, which is a measure of the non-reversibility of a Finsler metric. 
 \label{ejemplononreversibledynamics}
 \end{ejemplo}
\begin{ejemplo}
 The construction of {\it Example} \ref{ejemplononreversibledynamics} in the case when $(M,F)$ is a Riemannian structure provides an example of reversible dynamics. In this case the dynamics is the geodesic flow and the formal limit
  \begin{align*}
  \lim_{X\to A}\frac{1}{d(A,X)}(\Omega(X)-\Omega(X))=0
    \end{align*}
    along $\Phi$ holds good.
    This is also the case for the more general case of {\it reversible Finsler metrics}.
\end{ejemplo}
 \begin{ejemplo}
 Let us consider a quantum system, where pure states are described by elements $|\Psi\rangle$ of a Hilbert space $\mathcal{H}$. In the Schr\"odinger picture of dynamics \cite{Dirac1958, Hatfield}, the evolution law is given by the relation
 \begin{align}
 \Phi_\mathcal{H}(t,|\Psi (0)\rangle ) =\,\mathcal{U}(t)|\Psi (0)\rangle,
 \label{Dynamical law in quantum mechanics}
 \end{align}
where $\mathcal{U}(t)$ is the evolution operator \cite{Dirac1958}. The configuration space is the manifold of real numbers $\mathbb{R}$. We consider the sheaf over $\mathbb{R}$ with stalk isomorphic to the Hilbert space $\mathcal{H}$. Thus the sheaf we consider is of the form $\mathcal{E}\equiv\,\mathcal{H}\times \mathbb{R}$. The sections of the sheaf correspond to {\it arbitrary time evolutions} $t\mapsto |\Psi(t)\rangle$, not only Schr\"odinger's time evolutions.

In this setting, let us consider first the function $\Omega$ is defined by the relation
\begin{align}
\Omega:\Gamma (\mathcal{H}\times \mathbb{R})\to \mathbb{C},\,\, |\Psi(0)\rangle\mapsto |\langle \Psi(0)|\Phi_\mathcal{H}(t,|\Psi (0)\rangle )\rangle  |^2=\,|\langle\Psi(0)|\mathcal{U}(t)|\Psi (0)\rangle |^2.
\label{Omegaquantum}
\end{align}
If the Hamiltonian $\hat{H}$ of the evolution is Hermitian, then the conjugate dynamics $\Phi^c_\mathcal{H}$ is given by the expression
\begin{align*}
 \Phi^c_\mathcal{H}(t,|\Psi (0)\rangle )=\,\mathcal{U}^\dag (t)|\Psi (0)\rangle
\end{align*}
and function $\Xi_\Omega$ as defined in \eqref{definitionofnonreversibilityfunction} can be shown is trivially zero,
\begin{align*}
&\lim_{t\to\,0}\,\frac{1}{t}\,\left(|\langle \Psi(0)|\Phi_\mathcal{H}(t,|\Psi (0)\rangle )\rangle  |^2 -\,|\langle \Psi(0)|\Phi^c_\mathcal{H}(t,|\Psi (0)\rangle )\rangle  |^2\right)\\
& =\,\lim_{t\to\,0}\,\frac{1}{t}\,\left(|\langle\Psi(0)|\mathcal{U}(t)|\Psi (0)\rangle|^2 -\,|\langle\Psi(0)|\mathcal{U}^\dag(t)|\Psi (0)\rangle|^2\right)=0.
\end{align*}
Therefore, if we choose $\Omega$ as given by \eqref{Omegaquantum}, then unitary quantum mechanical evolution implies that the function $\Xi_\Omega= 0$. This choice for $\Omega$ is natural, since \eqref{Omegaquantum} is the probability transition for the possible evolution from the initial state $|\Psi (0)\rangle$ to itself by the $\mathcal{U}(t)$ evolution.

The condition $\Xi_\Omega =0$  does not imply that the dynamics $ \Phi_\mathcal{H}$ is reversible, since it could happen that the choice of another function $\widetilde{\Omega}$ is such that $\Xi_{\widetilde{\Omega}}\neq 0$. Indeed, let us consider instead the function
\begin{align}
\widetilde\Omega:\Gamma (\mathcal{H}\times \mathbb{R})\to \mathbb{C},\quad  |\Psi(0)\rangle\mapsto \langle \Psi(0)|\Phi_\mathcal{H}(t,|\Psi (0)\rangle )\rangle  =\,\langle\Psi(0)|\mathcal{U}(t)|\Psi (0)\rangle .
\label{Omegaquantum2}
\end{align}
For the Schr\"odinger evolution
\begin{align}
\imath\,\frac{d}{dt}\left(\mathcal{U}(t)|\psi(0)\rangle\right) =\,\hat{H}\,\mathcal{U}(t)|\psi(0)\rangle
\end{align}
\begin{align}
\Xi_{\widetilde{\Omega}}=\,2\,\imath\,\,\langle \Psi(0)|\hat{H}|\Psi(0)\rangle,
\end{align}
which is in general non-zero. Therefore, the dynamics of a quantum system is in general non-reversible.
\end{ejemplo}

\begin{ejemplo}The non-reversibility of quantum processes is usually related to the notion of entropy. In quantum theory, here are several notions of entropy, but the one relevant for us here is the notion that emerges in scattering theory. Indeed,
a version of the $H$-theorem for unitary quantum dynamics can be found in the textbook from S. Weinberg, \cite{Weinberg1995} section 6.6. The entropy function is defined by the expression
\begin{align}
S\,:=\,-\int \,d\alpha\,P_\alpha\,\ln\,(P_\alpha/c_\alpha),
\label{H-theorem}
\end{align}
where in scattering theory $P_\alpha \,d\alpha$ is the probability to find the state in a volume $d\alpha$ along the quantum state $|\Psi_\alpha\rangle$ and $c_\alpha$ is a normalization constant.
The change of entropy can be determined  using scattering theory,
\begin{align}
\nonumber -\frac{d}{dt}\left\{\int \,d\alpha\,P_\alpha\,\ln\,(P_\alpha/c_\alpha)\right\}=\, -\int d\alpha\,&\int d\beta\,\left(1+\,\ln (P_\alpha/c_\alpha)\right)\\
&\left( P_\beta\,\frac{d\Gamma(\beta\to \alpha)}{d\alpha}-\,P_\alpha\,\frac{d\Gamma(\alpha\to\beta)}{d\beta}\right).
\label{derivative S}
\end{align}
As a consequence of the unitary property of the $S$-matrix, it can be shown that the change with time of entropy \eqref{H-theorem} is not decreasing.

 Although appealing, one can cast some doubts that $dS/dt$ can be interpreted as  a function of the form $\Xi_\Omega$. In particular, the first remark is about the comparison between the formal expressions \eqref{definitionofnonreversibilityfunction} for $\Xi_\Omega$ and \eqref{derivative S} for the derivative $dS/dt$. The difficulty in this formal identification relies on the factor $\left(1+\,\ln (P_\alpha/c_\alpha)\right)$, which is a short of non-symmetric factor.

  The second difficulty in the formal identification of \eqref{definitionofnonreversibilityfunction} with the derivative \eqref{derivative S} is at the interpretational level. First, in order to interpret the transitions amplitudes as measurable decay rates, a foreign macroscopic arrow of time must be included. Second, the interpretation of the relation \eqref{H-theorem} is applied to  ensembles of identically prepared systems, indicating that beneath this {\it H-theorem} there is indeed an statistical interpretation of the scattering amplitude and cannot be applied to the detailed dynamics of an individual quantum system.
\end{ejemplo}
\begin{ejemplo}
Let us consider the case of physical systems where the weak sector of the electro-weak interaction of the Standard Model of particle physics is involved. The weak interaction slightly violates the $CP$-symmetry, as it is demonstrated in experiments measuring the decay rates of the $K^0$-$\bar{K^0}$ systems. Assuming that the $CPT$-theorem of relativistic quantum field theory holds \cite{Weinberg1995}, since such experiments show a violation of the $CP$-symmetry,  then the $T$-symmetry invariance of the $S$-matrix must be (slightly) violated, a fact that has been experimentally observed.

However, this non-reversibility character of the weak interaction is an indirect one and does not correspond to our notion of non-reversible dynamical law. It relies on an external element of the dynamics, namely, the existence of an external observer with a notion of macroscopic time. These elements are foreign to our notion of non-reversible dynamics, definition \ref{nonreversibledynamics}. General speaking, invariance under $T$-parity and reversibility as formulated in this paper are different notions.
 \end{ejemplo}

\section{Local dynamical time arrow associated to a non-reversible dynamics}
 The notion of reversible and non-reversible dynamical law that we are discussing is attached to the existence of a function $\Omega:\Gamma\mathcal{E}\to \mathbb{K}$ such that the property \eqref{definitionofnonreversibilityfunction} holds good. If one such function $\Omega$ is found, then there is a dynamical arrow of time defined by the following criteria:
 \begin{definicion}
 Given a non-reversible dynamics $\Phi_\mathcal{E}:J\times \Gamma\mathcal{E}\to \Gamma\mathcal{E}$ and a function $\Omega:\Gamma \mathcal{E}\to \mathbb{K}$ such that the corresponding non-reversibility $\Xi_\Omega$ is different from zero, a {\it global dynamical arrow of time} is a global choice on a sheaf $\mathcal{E}$ of a non-zero function $\Xi_\Omega$.
 \end{definicion}
In the case when there is a measure involved, it is only required that the non-reversibility function $\Xi_\Omega$ be different from zero for many evolutions.

 Given a non-reversible local dynamics as in {\it definition} \ref{nonreversibledynamics}, the sign of the function \eqref{definitionofnonreversibilityfunction} is well defined at least on local domains of $J\times \mathcal{E}$. Therefore, when the above definition holds at a local level, then one can speak of a {\it local dynamical arrow of time}. The existence of a dynamical arrow of time with constant sign defined in the whole configuration space $\mathcal{M}$ is a non-trivial requirement. For example, in the case of Finsler structures, it is not possible, in general, to keep the sign of the difference \eqref{FmenosF} constant  on the whole tangent manifold $TM$. However, the difference \eqref{FmenosF} is a continuous function on $TM$. Hence one can reduce the domain of definition of the local dynamical time arrow to the open set where the difference is positive. It is always possible to find a local dynamical arrow of time.

 The notion of fundamental time arrow as discussed above immediately raises the question of the relation with the entropic or thermodynamic time arrow, indicated by the non-decrease of the entropy function. For a general local dynamics, a notion of time arrow based upon the non-reversible dynamics and the notion of time arrow based upon the increase of an entropy function will in general not coincide, since the former can change direction after crossing the condition $\Xi_\Omega =0$. In the case when $\dim (\mathbb{K})=1$, if $\Xi_\Omega  \neq 0$,  the {\it turning points} are among sections such that $\Xi_\Omega =0$. They correspond to points where the time arrow associated with a non-reversible local dynamics can change sign respect to the time arrow based on the evolution of the entropy function, which is always non-decreasing for systems and for arbitrary thermodynamic sub-systems \cite{Kondepudi Prigogine 2015}.

  Another difference between the dynamical arrow of time and the thermodynamic arrow of time appears if $\dim(\mathbb{K})> 1$. Then there is no prescribed way to associate the arrow of time of a non-reversible dynamics with the arrow of time of an entropy function.

  The above arguments motivate the local identification of the entropic arrow of time with a local arrow of time based upon {\it definition} \ref{nonreversibledynamics},
  \begin{proposicion}
  For any local dynamical law $\Phi_\mathcal{E}:J\times \Gamma\mathcal{E}\to \Gamma\mathcal{E}$ where $J\subset \,\mathbb{K}$, if there is an entropy function defined on $\mathcal{E}$, then the dynamical arrow of time associated to $\Phi_\mathcal{E}$ and a non-zero non-reversibility function $\Xi_\Omega$ coincide locally with the entropic arrow of time signaled by the generalized second principle of thermodynamics.
  \end{proposicion}

\section{Conclusion: Supporting arguments in favour of the  non-reversibility property of the fundamental dynamics}
After the development of a theory of dynamics along the above lines, we are in position to address with some detail the question that partially motivated the present research.
Three independent arguments support the non-reversibility character of any hypothetical fundamental dynamics.
\\
{\bf 1. Non-reversibility is generic}. Given any configuration space $\mathcal{M}$, the collection of non-reversible dynamics as discussed above is larger and contains the set of reversible ones. Let us consider reversible dynamics $\Phi_r$. Then for any function $\Omega:\Gamma\mathcal{E}\to\,\mathbb{K}$ we have that
     \begin{align}
  \lim_{t\to\,0}\,\frac{1}{t}\,\left(\Omega(\Phi_r(t,A))-\Omega(\Phi^c_r(t,A))\right)=0.
  \label{reversibility condition}
    \end{align}
    But this condition is easy to be spoiled. Almost any deformation $\widetilde{\Phi}$ of a reversible dynamics $\Phi_r$ will admit a function $\widetilde{\Omega}$ such that
     \begin{align*}
      \lim_{t\to\,0}\,\frac{1}{t}\,\left(\widetilde{\Omega}(\widetilde{\Phi}_\mathcal{E}(t,A))-\widetilde{\Omega}(\widetilde{\Phi}^c_\mathcal{E}(t,A))\right)\neq 0.
     \end{align*}

Given a non-reversible dynamics, an associated reversible dynamics can be constructed by a process of {\it time symmetrization} as discussed before.  Such a process is information loss, since for many non-reversible dynamics have associated the same reversible dynamics. Therefore, the {\it averaged reversible dynamics} determines a class of equivalence of non-reversible dynamics. Conversely, any reversible dynamics can be obtained from a (non-unique) non-reversible dynamics by a process time symmetrization. Furthermore, the notion of symmetrization  can be applied to both, continuous or discrete  dynamics. 
All these suggest that the category of non-reversible dynamics is the natural choice to formulate a fundamental dynamics.
\\
{\bf 2. Finsler structures as models for non-reversible dynamical systems}. One of the important characteristics of Finsler geometry is its ubiquity in the category of differentiable manifolds and differentiable maps. Finsler structures are natural objects in the sense that they can be defined on any manifold ${M}$ with enough regularity and with some few additional natural conditions. A Randers type metric is the prototype of non-reversible Finsler structure. It is a small perturbation of a Riemannian structure \cite{Randers,BaoChernShen}. Any Haussdorff, paracompact manifold $M$ admits a Riemannian structure \cite{Warner} and if  additional conditions to ensure the existence of globally defined, globally bounded, smooth vector fields are imposed, then Randers metrics can be constructed globally on the manifold.  Randers type structure on large dimensional tangent spaces $TM_t$ of large product manifolds $M_t=\,\prod^N_{k=1} \times M_k$ are the geometric structure for the fundamental dynamics in emergent quantum mechanics \cite{Ricardo06,Ricardo2014}.

In the category of Finsler structures, there is a natural process of averaging \cite{Ricardo05}. Such geometric procedure is information loss, since many different Finsler metrics have associated the same Riemannian metric under averaging. Furthermore, the averaging could be mimicked as a geometric flow on the tangent space ${TM}$. Such hypothetical geometric flow determines a non-reversible dynamics on $TM$. The result of this geometric flow  is a Riemannian structure on $M$. In emergent quantum mechanics, such a flow as used to define an emergent time and emergent arrow of time \cite{Ricardo06,Ricardo2014}.

The extension of the averaging flow to metric structures with indefinite signature is still an open  problem.  However, progress has been made in the case when the spacetime admits a strong form of time orientability \cite{Ricardo 2021a}. It is interesting to show the potential relation of such strong time orientation with a definition of arrow of time, that is, to explore the possible dynamical origin of the strong time orientation.
\\
{\bf 3. Quantum scattering theory already contains an $H$-theorem}. The above version of the $H$-theorem in scattering theory shows how  unitary quantum local dynamics provides a mathematical entropy function \eqref{H-theorem}. Furthermore, the proof of the theorem does not rely on Born's approximations \cite{Weinberg1995}, providing to the result a general character. Therefore, the scattering $H$-theorem can be extended to other deterministic dynamics via Koopman-von Neumann theory, despite the interpretational issues on the lack of a perfectly close interpretation of the entropy function as a non-reversibility function.

The above general remarks lead to the conclusion that, despite the time-invariance reversibility character of most quantum dynamics, a fundamental theory based upon Koopman-von Neumann formalism must be non-reversible. Furthermore, the corresponding (local) time arrow should coincide (locally) with the quantum time arrow associated with the $H$-theorem of scattering theory. The reversible character of current quantum and classical field theories could be interpreted as an emergent property, according to emergent quantum mechanics \cite{Ricardo2014}.

\end{document}